\documentclass[11pt,letterpaper,preprintnumbers,superscriptaddress]{revtex4}
\usepackage{amsfonts,amsmath,amsthm,amssymb,amscd,dsfont}
\usepackage{color,graphicx,tikz,array}
\usepackage{hyperref,cleveref}

\normalfont
\usepackage[T1]{fontenc}
\usepackage[linesnumbered,algoruled]{algorithm2e}

\newtheorem{theorem}{Theorem}
\newtheorem{lemma}[theorem]{Lemma}
\newtheorem{proposition}[theorem]{Proposition}

\crefname{inequality}{ineq.}{ineqs.}
\creflabelformat{inequality}{(#2#1#3)}
\crefname{definition}{}{}
\creflabelformat{definition}{(#2#1#3)}

\begin{document}

\title[A regional belief propagation algorithm]{A belief propagation algorithm based on domain decomposition} 
\date{\today}

\author{Brad~Lackey}
\affiliation{Joint Center for Quantum Information and Computer Science, University of Maryland, College Park}
\affiliation{Institute for Advanced Computer Studies, University of Maryland College Park}
\affiliation{Departments of Computer Science and Mathematics, University of Maryland, College Park}
\affiliation{The Cryptography Office, Mathematics Research Group, National Security Agency}


\begin{abstract}
This note provides a detailed description and derivation of the domain decomposition algorithm that appears in previous works by the author. Given a large re-estimation problem, domain decomposition provides an iterative method for assembling Boltzmann distributions associated to small subproblems into an approximation of the Bayesian posterior of the whole problem. The algorithm is amenable to using Boltzmann sampling to approximate these Boltzmann distributions. In previous work, we have shown the capability of heuristic versions of this algorithm to solve LDPC decoding and circuit fault diagnosis problems too large to fit on quantum annealing hardware used for sampling. Here, we rigorously prove soundness of the method.
\end{abstract}

\maketitle

\section{Introduction}

In our previous work \cite{bian2014discrete, bian2016mapping}, we explored methods to decompose constrained optimization problems into subproblems that could be solved on quantum annealing hardware, where the individual solutions can then be reassembled into a solution of the whole problem. We have used the term \emph{domain decomposition} to represent any of a number of heuristic belief propagation-like algorithms to do this. In \cite{bian2014discrete}, we created a variant of the min-sum algorithm to solve LDPC decoding problems with more variables than qubits available on the available hardware. In \cite{bian2016mapping}, we created a sum-product style algorithm that solved fault diagnosis problems too large to fit on the hardware directly.

While ultimately successful in solving these problems, the versions of the domain decomposition method used in those works was heuristic: there was no guarantee that the object formed by reassembling the solutions on each subproblem was in any way related to the original problem. The goal of this note is provide a mathematical proof of the soundness of a specific domain decomposition algorithm (Algorithm \ref{algorithm}) formed by balancing ``free energy flow'' between subproblems. In \S{2} we review the well known relationship between the Boltzmann distribution, Helmholtz free energy, and belief propagation algorithm. In \S{3} we develop the notion of a regional approximate free energy analogous to the Bethe free energy. In \S{4} we examine the free energy of a single region and show that when the free energy between regions is balanced, a critical point for the regional approximate free energy is obtained (Theorem \ref{theorem:reconstitute}). We then given an explicit statement for the domain decomposition algorithm (Algorithm \ref{algorithm}) and prove that a stationary point for this algorithm is such a critical point (Theorem \ref{theorem:soundness}).

Like belief propagation there is no guarantee that the domain decomposition algorithm converges. In the case of belief propagation failure to converge is often attributed to a lack of convexity of the Bethe approximate free energy stemming from loops in the factor graph. Several sufficient conditions for convergence have been based on this observation \cite{tatikonda2002loopy, heskes2003stable, heskes2004uniqueness, mooij2007sufficient, watanabe2009graph, watanabe2011uniqueness}, and other forms of generalized belief propagation algorithms have been developed from convex bounds or treating convex and concave regions separately \cite{heskes2002approximate, teh2002unified, yuille2002cccp, heskes2003approximate, globerson2007approximate}. While the algorithm we propose here is a significant departure from these methods, it is possible that one or more of these techniques could apply. We leave such explorations for future work.

Algorithm \ref{algorithm} relies on obtaining marginals of Boltzmann distributions along certain ``boundary'' variables. While the intent was to empirically estimate these using a dedicated device that produces Boltzmann samples, these marginals could also be computed by any number of other means such as simulated annealing or even belief propagation. It would be interesting to see if domain decomposition using---say---simulated annealing to compute these marginals could outperform belief propagation by carefully designing regions that internalize many of the loops in the factor graph, and so reduce the nonconvexity of the free energy approximation.

\section{Free energy and belief propagation}

Belief propagation \cite{pearl2014probabilistic} is an iterative method to approximate posterior marginal distributions of a constrained optimization problem given observations. The most common version of belief propagation, the sum-product algorithm, can derived in a straightforward manner by carefully tracking conditional probabilities and making suitable independence assumptions; for a concise presentation see \cite{kschischang2001factor}. Such algorithms, including various variants of belief propagation, can be viewed as a type of generalized distributive law on semirings \cite{aji2000generalized}.

Our starting point is the variational approach to Bayesian re-estimation, \cite{yedidia2000generalized, aji2001generalized, yedidia2001bethe, pakzad2002belief, yedidia2003understanding, yedidia2005constructing}: the Bayesian posterior distribution is a critical point of a Helmholtz free energy constructed from the given problem. Specifically, given a space $\Omega$ of configurations, let us write our target probability as a Gibbs state based on the product of problem dependent factors
\begin{equation}\label{equation:ansatz}
p_0(x) = \frac{1}{Z} \prod_{\alpha=1}^m f_\alpha(x) = \frac{1}{Z(T)} e^{-\sum_{\alpha=1}^m E_\alpha(x)/kT},
\end{equation}
where each factor term $f_\alpha(x) \propto e^{-E_\alpha(x)/kT}$ contributes an energy $E_\alpha(x)$ to the total energy $E(x) = \sum_{\alpha=1}^m E_\alpha(x)$. The Helmholtz free energy of an ensemble with probability distribution $p$ is
$$A(p) = U(p) - TH(p) = \sum_{x\in\Omega} p(x)E(x) + kT\sum_{x\in\Omega} p(x)\log p(x).$$
From (\ref{equation:ansatz}) above one can write
$$E(x) = \sum_{\alpha=1}^m E_\alpha(x) = -kT(\log p_0(x) + \log Z(T)).$$
Inserting this into the above equation,
\begin{eqnarray*}
A(p) &=& \sum_{x\in\Omega} \left( -kT p(x) \log p_0(x) + kT p(x) \log p(x) \right) - kT \log Z(T)\\
&=& -kT\log Z(T) + kT\sum_x p(x)\log\left(\frac{p(x)}{p_0(x)}\right)\\
&=& -kT\log Z(T) + kTD_{KL}(p \|p_0).
\end{eqnarray*}
The Kullback-Leibler divergence has $D_{KL}(p \|p_0) \geq 0$, with equality if and only $p = p_0$, and therefore the posterior is the global minimum of this Helmholtz free energy.

Decomposing the Helmholtz free energy as
$$A(p) = U(p) - TH(p) = \sum_{x\in\Omega}\sum_{\alpha=1}^m p(x)E_\alpha(x) + kT\sum_{x\in\Omega} p(x)\log p(x),$$
one sees each factor contributes an energy term to the internal energy. Let us assume that the energies are ``local.'' We do not want to be too formal about this; simply, we assume the configurations in our space all have the form $x = (x_1, \dots, x_n)$ and that each $E_\alpha$ does not depend on all the variables $x_1,\dots,x_n$, but rather only a few of them. Let us introduce the notation
\begin{itemize}
\item $j \prec \alpha$ to mean $E_\alpha$ (through $f_\alpha$) depends nontrivially on $x_j$,
\item $x_\alpha = (x_j\::\: j \prec \alpha)$ for the support of $E_\alpha$, and
\item $\Omega_\alpha$ for the domain of the vector $x_\alpha$.
\end{itemize}
The marginal distribution with respect to one of these supports is $b_\alpha(x_\alpha) = \sum_{x\setminus x_\alpha} p(x)$ (the notation $b_\alpha$ refers to the ``belief'' and hence ``belief propagation''). The internal energy then has the form
$$\sum_{x\in\Omega} p(x) E_\alpha(x) = \sum_{x_\alpha\in \Omega_\alpha} b_\alpha(x_\alpha) E_\alpha(x_\alpha).$$
That is, the internal energy is linear in the sense that
$$U(p) = \sum_{x\in\Omega}\sum_{\alpha=1}^m p(x)E_\alpha(x) = \sum_{\alpha=1}^m \sum_{x_\alpha\in \Omega_\alpha} b_\alpha(x_\alpha) E_\alpha(x_\alpha) = \sum_{\alpha=1}^m U_\alpha(b_\alpha).$$

At this point it is natural to disassociate the beliefs $b_\alpha$ with the marginals of $p$ and pose a local free energy associated to each factor (for which we will reuse the name $b_\alpha$ for its argument)
$$A_\alpha(b_\alpha) = U_\alpha(b_\alpha) - TH(b_\alpha)  = \sum_{x_\alpha} b_\alpha(x_\alpha) E_\alpha(x_\alpha) + kT\sum_{x_\alpha} b_\alpha(x_\alpha)\log b_\alpha(x_\alpha).$$
The main problem lies in ``localizing'' the entropy term. The sum of the factors' local free energies does not recover the Helmholtz free energy partially because of the nonlinearity inherent in the entropy, but mostly because one has grossly over-counted entropy contributions from factors sharing common variables. To illustrate this by a simple example, suppose $\Omega = \Omega_1\times\Omega_2\times\Omega_3$ and $p(x) = p(x_1,x_2,x_3)$ is uniform; form the marginals
$$b_1(x_2,x_3) = \sum_{x_1\in\Omega_1} p(x_1,x_2,x_3) \text{ and } b_2(x_1,x_3) = \sum_{x_2\in\Omega_2} p(x_1,x_2,x_3),$$
which are also uniform. Then
$$H(p) = k\log |\Omega| = k(\log |\Omega_1| + \log |\Omega_2| + \log |\Omega_3|).$$
and
$$H(b_1) = k(\log |\Omega_2| + \log |\Omega_3|) \text{ and } H(b_2) = k(\log |\Omega_1| + \log |\Omega_3|).$$
Therefore $H(b_1) + H(b_2)$ over-counts the entropy by $k\log |\Omega_3|$ simply because $x_3$ in the support of each belief.

The Bethe approximation overcomes this failure by correcting the entropy count at each variable. For each variable $x_j$ let us write $C_j = \#\{ \alpha : j \prec \alpha \}$, the number of factors involving this variable. Removing extra entropy through over counting gives
\begin{eqnarray}\nonumber
A_\text{Bethe} &=& \sum_{\alpha=1}^m\left[\,\sum_{x_\alpha\in\Omega_\alpha} b_\alpha(x_\alpha)E_\alpha(x_\alpha) + kT\sum_{x_\alpha\in\Omega_\alpha} b_\alpha(x_\alpha)\log b_\alpha(x_\alpha)\,\right]\\\label{eqn:Bethe}
&& \qquad +\ kT \sum_{j=1}^n \left((1-C_j)\cdot\sum_{x_j\in\Omega_j} b_j(x_j) \log b_j(x_j)\right).
\end{eqnarray}
There is some freedom to use different weights to correct the entropy contributions, which leads to variants of the sum-product algorithm \cite{wiegerinck2003fractional, weller2015bethe}. We will only consider the Bethe approximation as given above.

Here, the ``ensemble'' of the Bethe approximate free energy is a collection of beliefs $\{b_\alpha, b_j\}$, one for each factor and variable. The target posterior distribution minimizes the Helmholtz free energy, so it is not unreasonable to attempt to approximate this posterior with the minimum of the Bethe approximation. However, the minimum of the Helmholtz free energy is taken over global probability distributions, while the Bethe approximation is a function of disjoint beliefs. To rectify this, one adds consistency constraints on the beliefs so to make them marginals of a single distribution:
\begin{itemize}
\item $\sum_{x_\alpha} b_\alpha(x_\alpha) = 1$ and $\sum_{x_j} b_j(x_j) = 1$, and
\item whenever $j \prec \alpha$ we require $\sum_{x_\alpha\setminus x_j} b_\alpha(x_\alpha) = b_j(x_j)$.
\end{itemize}
Enforcing these conditions with Lagrange multipliers produces a constrained Bethe approximate free energy:
\begin{eqnarray*}
\tilde{A}_\text{Bethe} &=& \sum_{\alpha=1}^m\left[\,\sum_{x_\alpha\in\Omega_\alpha} b_\alpha(x_\alpha)E_\alpha(x_\alpha) + kT\sum_{x_\alpha\in\Omega_\alpha} b_\alpha(x_\alpha)\log b_\alpha(x_\alpha)\,\right]\\
&& \quad +\ kT \sum_{j=1}^n \left((1-C_j)\cdot\sum_{x_j\in\Omega_j} b_j(x_j) \log b_j(x_j)\right) \\
&& \quad +\  \sum_{\alpha=1}^m \lambda_\alpha \left(\sum_{x_\alpha\in\Omega_\alpha} b_\alpha(x_\alpha) - 1\right)
+ \sum_{j=1}^n \lambda_j \left(\sum_{x_j\in\Omega_j} b_j(x_j) - 1\right)\\
&& \quad +\  \sum_{\alpha=1}^m \sum_{j\prec\alpha} \sum_{x_j\in\Omega_j}\lambda_{j\alpha}(x_j) \left(\sum_{x_\alpha \setminus x_j} b_\alpha(x_\alpha) - b_j(x_j)\right).
\end{eqnarray*}

We find relations at interior critical points of $\tilde{A}_\text{Bethe}$ by setting various derivatives to zero. The derivatives with respect to the multipliers simply recover the constraints when set to zero. The two nontrivial types of derivatives are
\begin{eqnarray*}
\frac{\partial \tilde{A}_\text{Bethe}}{\partial b_\alpha(x_\alpha)} &=&  E_\alpha(x_\alpha) +kT(\log b_\alpha(x_\alpha) + 1) + \lambda_\alpha + \sum_{j\prec\alpha} \lambda_{j\alpha}(x_j), \text{ and}\\
\frac{\partial \tilde{A}_\text{Bethe}}{\partial b_j(x_j)}  &=& kT(1-C_j) (\log b_j(x_j) +1) + \lambda_j - \sum_{\alpha \succ j} \lambda_{j\alpha}(x_j).
\end{eqnarray*}
Setting the first of these to zero produces the equation
$$b_\alpha(x_\alpha)= e^{-(1+\lambda_\alpha/kT)} e^{-E_\alpha(x_\alpha)/kT} \prod_{j\prec\alpha}e^{-\lambda_{j\alpha}(x_j)/kT}.$$
In particular, at a critical point the multiplier $\lambda_\alpha$ can be selected, and is completely determined by, the normalization constraint $\sum_{x_\alpha} b_\alpha(x_\alpha) = 1$. So, we are free to work with $b_\alpha$ unnormalized, for which we have
\begin{equation}\label{equation:BPfactor}
b_\alpha(x_\alpha) \propto f_\alpha(x_\alpha) \cdot \prod_{j\prec\alpha}e^{-\lambda_{j\alpha}(x_j)/kT}
\end{equation}
at any interior critical point. Similarly, the second type of derivative above produces the relation
\begin{equation}\label{equation:BPvariable}
b_j(x_j) \propto \prod_{\alpha \succ j} e^{-\lambda_{j\alpha}(x_j)/kT(C_j-1)}
\end{equation}
at an interior critical point, where now $\lambda_j$ is selected, and determined by, the normalization of $b_j$.

The remaining constraints are the consistency of the marginals. Computing the marginal $\sum_{x_\alpha\setminus x_j} b_\alpha(x_\alpha) = b_j(x_j)$ using (\ref{equation:BPfactor}) one finds the additional relations on each $b_j$:
$$b_j(x_j) \propto e^{-\lambda_{j\alpha}(x_j)/kT}\cdot \sum_{x_\alpha\setminus x_j} f_\alpha(x_\alpha) \prod_{r\prec\alpha\setminus j} e^{-\lambda_{r\alpha}(x_j)/kT}.$$
We define the function
$$M_{\alpha \to j}(x_j) \propto \sum_{x_\alpha\setminus x_j} f_\alpha(x_\alpha) \prod_{r\prec\alpha\setminus j} e^{-\lambda_{r\alpha}(x_j)/kT},$$
with the condition $\sum_{x_j} M_{\alpha \to j}(x_j) = 1$. One can simplify the above relation to
\begin{equation}\label{equation:BPinter}
b_j(x_j) \propto e^{-\lambda_{j\alpha}(x_j)/kT} \cdot M_{\alpha \to j}(x_j).
\end{equation}

Note the trivial equation 
$$\prod_{\beta \succ j\::\: \beta \not= \alpha} b_j(x_j) = b_j(x_j)^{C_j-1},$$
which is just a restatement of the definition of $C_j$. Now, we evaluate the left side of this equation using (\ref{equation:BPinter}) and the right side using (\ref{equation:BPvariable}). This results in the relation
$$\prod_{\beta \succ j\::\: \beta \not= \alpha} e^{-\lambda_{j\beta}(x_j)/kT} \cdot M_{\beta \to j}(x_j) \propto \prod_{\beta \succ j} e^{-\lambda_{j\beta}(x_j)/kT}.$$
Canceling common terms from both side leaves
\begin{equation}\label{equation:BPrel}
e^{-\lambda_{j\alpha}(x_j)/kT} \propto \prod_{\beta \succ j\::\: \beta \not= \alpha} M_{\beta \to j}(x_j).
\end{equation}
Now we define the function
$$M_{j\to \alpha}(x_j) \propto  e^{-\lambda_{j\alpha}(x_j)/kT},$$
where as before we require $M_{j\to \alpha}(x_j)$ to be a probability distribution.

\begin{proposition}
Suppose that for each variable $x_j$ and factor $f_\alpha$ with $j \prec \alpha$ one has two probability distributions, $M_{\alpha \to j}(x_j)$ and $M_{j \to \alpha}(x_j)$, which jointly satisfy
\begin{eqnarray*}
M_{\alpha \to j}(x_j) &\propto& \sum_{x_\alpha\setminus x_j} f_\alpha(x_\alpha) \prod_{k\prec\alpha\setminus j} M_{k \to \alpha}(x_k), \text{ and}\\
M_{j\to \alpha}(x_j) &\propto& \prod_{\beta \succ j\::\: \beta \not= \alpha} M_{\beta \to j}(x_j).
\end{eqnarray*}
Then the probability distributions $b_\alpha(x_\alpha)$ and $b_j(x_j)$ satisfying
\begin{eqnarray*}
b_\alpha(x_\alpha) &\propto& f_\alpha(x_\alpha) \cdot \prod_{j \prec \alpha} M_{j \to \alpha}(x_j), \text{ and}\\
b_j(x_j) &\propto& \prod_{\alpha \succ j} M_{\alpha\to j}(x_j),
\end{eqnarray*}
are critical points of the Bethe approximate free energy with 
$$\sum_{x_\alpha\setminus x_j} b_\alpha(x_\alpha) = b_j(x_j), \text{ whenever $j \prec \alpha$.}$$
\end{proposition}

This result hands us the sum-product algorithm. We initialize distributions $M^{(0)}_{\alpha \to j}(x_j)$ and $M^{(0)}_{j \to \alpha}(x_j)$ in a reasonable way (which will be problem dependent) and iteratively redefine these as in the proposition:
\begin{eqnarray*}
M^{(t+1)}_{\alpha \to j}(x_j) &\propto& \sum_{x_\alpha\setminus x_j} f_\alpha(x_\alpha) \prod_{k\prec\alpha\setminus j} M^{(t)}_{k \to \alpha}(x_k), \text{ and}\\
M^{(t+1)}_{j\to \alpha}(x_j) &\propto& \prod_{\beta \succ j\::\: \beta \not= \alpha} M^{(t+1)}_{\beta \to j}(x_j).
\end{eqnarray*}
The proposition shows that if this converges to a stationary point, that point is a critical point of the Bethe approximate free energy and so may be taken as an approximation of the Bayesian posterior.

In practice, one has a class of factors that depend only a single variable $f_\alpha = f_\alpha(x_j)$; these typically arise as prior probabilities in a re-estimation problem. Let us write $\alpha = \{j\}$ in this case. Since $\{j\}\setminus j$ is vacuous, we have
$$M_{\{j\}\to j}^{(t)}(x_j) \propto f_{\{j\}}(x_j)$$
for all $t$ and hence there is no need to compute these. One can then simplify
$$M_{j\to\alpha}^{(t+1)} \propto f_{\{j\}}(x_j) \cdot \prod_{\beta\succ j\::\:\beta\not=\alpha,\{j\}}M^{(t+1)}_{\beta \to j}(x_j).$$
Since $j\not\prec \{k\}$ for $j \not= k$, one can also eliminate all factors arising from priors and use this last rule as the update rule at such nodes.

In most problems, one is primarily interested in the marginals $b_j(x_j)$, as these are the posterior probabilities one wishes to compute by re-estimation. From (\ref{equation:BPvariable}) these satisfy
$$b_j(x_j) \propto \prod_{\alpha\succ j} e^{-\lambda_{j\alpha}(x_j)/kT(C_j-1)} \propto e^{-(\sum_{\alpha\succ j}(\lambda_{j\alpha}(x_j)/(C_j-1)))/kT}.$$
That is, $b_j$ is the Gibbs state associated to the energy spectrum given by normalized Lagrange multipliers $\lambda_{j\alpha}(x_j)/(C_j-1)$. In particular, the posterior marginal is itself the minimum of the Helmholtz free energy
$$A_j(b_j) = \sum_{\alpha\succ j} \sum_{x_j\in\Omega_j} b_j(x_j)\lambda_{j\alpha}(x_j)/(C_j-1) + kT\sum_{x_j\in\Omega_j}b_j(x_j)\log b_j(x_j).$$

\section{Regional approximate free energies}

The presentation of the Bethe approximate free energy of the previous section was a purposefully unusual. Let us return to this construction, but rather than focus on the energies of single factors let us construct ``regions'' containing multiple factors. We introduce the notation $f_R(x) = \prod_{\alpha \in R} f_\alpha(x)$ for a region of factors $R$. Without loss of generality, we will assume that to every variable $x_j$ there is a factor $f_j(x_j)$. As noted at the end of the previous section, we can treat priors differently and so not include them in any of the regions. To be precise, all the factors that are not priors are partitioned into regions and the prior factors are treated individually. We hasten to indicate that unlike the usual approach to regional belief propagation algorithms \cite{yedidia2000generalized, aji2001generalized, pakzad2002belief}, we do not use a hierarchy of regions and the Kukuchi free energy approximation. Here we simply organize our factors in a different way and use the Bethe approximation above.

We extend the notation to $x_R$, meaning the collection of variables for which $f_R$ depends nontrivially; we also write $x_j \in x_R$ or $j \prec R$ when $x_j$ is one of these variables. Each region has internal energy
$$U_R(b_R) = \sum_{x_R} b_R(x_R)E_R(x_R),$$
where we have defined $E_R(x) = \sum_{\alpha \in R} E_\alpha(x)$. Similarly, if $E_j(x_j)$ is the energy associated to the prior $f_j(x_j)$, the internal energy at the variable $x_j$ is
$$U_j(b_j) = \sum_{x_j} b_j(x_j)E_j(x_j).$$
When $b_R$ and $b_j$ are obtained by marginalizing a global probability $p$, the internal energy of the whole system is
$$U(p) = \sum_R U_R(b_R) + \sum_j U_j(b_j).$$

Exactly at in the Bethe method, define the weight
$$C_j = \#\{R \::\: R \succ j\},$$
and the constrained regional approximate free energy:
\begin{eqnarray*}
\tilde{A}_\text{regional} &=& \sum_R\left(\sum_{x_R} b_R(x_R) E_R(x_R) + kT \sum_{x_R} b_R(x_R)\log b_R(x_R)\right)\\
&&\quad +\ \sum_j\left(\sum_{x_j} b_j(x_j) E_j(x_j) - kT(C_j-1) \sum_{x_j} b_j(x_j)\log b_j(x_j)\right)\\
&&\quad +\ \sum_R \lambda_R \left(\sum_{x_R} b_R(x_R) - 1\right) + \sum_j \lambda_j \left(\sum_{x_j} b_j(x_j) - 1 \right)\\
&&\quad +\ \sum_R \sum_{j\prec R} \sum_{x_j} \lambda_{jR}(x_j) \left( \sum_{x_R\setminus x_j} b_R(x_R) - b_j(x_j) \right).
\end{eqnarray*}

By a similar argument, setting the derivative with respect to $b_R(x_R)$ to zero yields
\begin{equation}\label{equation:GBPfactor}
b_R(x_R) \propto f_R(x_R) \cdot \prod_{j\prec R}e^{-\lambda_{jR}(x_j)/kT}.
\end{equation}

One finds the analysis of the derivative with respect to $b_j(x_j)$ splits into two cases. First if $C_j = 1$ (when the variable is internal to one region) one obtains
$$\frac{\partial\tilde{A}_\text{regional}}{\partial b_j(x_j)} = E_j(x_j) + \lambda_j - \lambda_{jR}(x_j),$$
where $R$ is the region that contains $j$. By setting this to zero and exponentiating, we find
$$e^{-\lambda_{jR}(x_j)/kT} \propto f_j(x_j).$$
However for variables contained in multiple regions (which we call \emph{boundary} variables) one finds
$$\frac{\partial \tilde{A}_\text{regional}}{\partial b_j(x_j)} = E_j(x_j) - (C_j-1)kT (\log b_j(x_j) + 1) +  \lambda_j - \lambda_{jR}(x_j).$$
Setting this derivative to zero produces the relation
\begin{equation}\label{equation:GBPvariable}
b_j(x_j) \propto (f_j(x_j))^{-1/(C_j-1)}\cdot \prod_{R\succ j} e^{-\lambda_{jR}(x_j)/kT(C_j-1)}.
\end{equation}

The variables forming the support of a region $R$ divide into boundary variables (denoted $\partial R$) and the remainder which we call interior (denoted $R^\circ$). Computing the marginal of the belief in (\ref{equation:GBPfactor}) at a boundary variable $j \in \partial R$ yields the relation
$$b_j(x_j) \propto e^{-\lambda_{jR}(x_j)/kT} \sum_{x_R\setminus x_j} f_{R}(x_R)
\prod_{r \prec R\setminus j} e^{-\lambda_{rR}(x_r)/kT}.$$
Note that in this product, $r \prec R\setminus j$ runs over all variables. For variables in $R^\circ$ we know $e^{-\lambda_{rR}(x_r)/kT} = f_r(x_r)$. Based on this we define the \emph{augmented} regional factor as the usual factor for that region times the priors for all its internal variables:
$$\tilde{f}_R(x_R) = f_R(x_R)\cdot \prod_{r \prec R^\circ} f_r(x_R).$$
As $j \prec \partial R$,
$$b_j(x_j) \propto e^{-\lambda_{jR}(x_j)/kT} \sum_{x_R\setminus x_j} \tilde{f}_{R}(x_R)
\prod_{r \prec \partial R\setminus j} e^{-\lambda_{rR}(x_r)/kT}.$$
This leads us to define
$$M_{R\to j}(x_j) = \sum_{x_R\setminus x_j} \tilde{f}_{R}(x_R) \prod_{r \prec \partial R\setminus j} e^{-\lambda_{rR}(x_r)/kT}$$
so that
\begin{equation}\label{equation:GBPinter}
b_j(x_j) \propto e^{-\lambda_{jR}(x_j)/kT} M_{R\to j}(x_j).
\end{equation}

Except for the fact that one deals only with messages to and from boundary variables, the arguments of the previous section apply. For a fixed region $R$ one has
$$\prod_{S \succ j \::\: S \not= R} b_j(x_j) = b_j(x_j)^{C_j-1}.$$
Evaluating the left and right sides with equations (\ref{equation:GBPinter}) and (\ref{equation:GBPvariable}) respectively gives the relation
$$\prod_{S\succ j \::\: S \not= R} e^{-\lambda_{jS}(x_j)/kT} M_{S\to j}(x_j) \propto
f_j(x_j)^{-1} \cdot \prod_{S\succ j} e^{-\lambda_{jS}(x_j)/kT}.$$
Canceling common terms produces the analogous result:
$$M_{j \to R}(x_j) \propto e^{-\lambda_{jR}(x_j)/kT} \propto f_j(x_j)\cdot \prod_{S\succ j \::\: S \not= R} M_{R\to j}(x_j).$$

\begin{proposition}\label{proposition:regional-BP}
Suppose that for each region $R$ and variable $x_j$ with $j \prec \partial R$ one has two probability distributions, $M_{R \to j}(x_j)$ and $M_{j \to R}(x_j)$, which jointly satisfy
\begin{eqnarray*}
M_{R \to j}(x_j) &\propto& \sum_{x_R\setminus x_j} \tilde{f}_{R}(x_R) \prod_{r\prec \partial R\setminus j} M_{r \to R}(x_r), \text{ and}\\
M_{j\to R}(x_j) &\propto& f_j(x_j)\cdot\prod_{S \succ j\::\: S \not= R} M_{S \to j}(x_j).
\end{eqnarray*}
Then the probability distributions $b_R(x_R)$ and $b_j(x_j)$ satisfying
\begin{eqnarray*}
b_R(x_R) &\propto& \tilde{f}_R(x_R) \cdot \prod_{j \prec \partial R} M_{j \to R}(x_j), \text{ and}\\
b_j(x_j) &\propto& f_j(x_j)\cdot \prod_{R \succ j} M_{R\to j}(x_j),
\end{eqnarray*}
are critical points of the regional approximate free energy with 
$$\sum_{x_R\setminus x_j} b_R(x_R) = b_j(x_j), \text{ whenever $j \prec \partial R$.}$$
\end{proposition}

Note that this proposition yields a generalized belief propagation algorithm. However this algorithm becomes impractical as the size of the regions increases. With large regions, there is potentially far fewer variables $x_j$ for which we need to compute messages. However the marginalization over $x_R\setminus x_j$ becomes intractable.

\section{Modified free energy of a single region}

In the previous section we developed a regional belief propagation algorithm that rapidly becomes inefficient as the size of regions get large. However, if we assume we have access to a device that samples from the Boltzmann distribution, can this capability aid us in computing---or at least in approximating---these marginals?

Given a regional decomposition as in the last section, the Helmholtz free energy of the factors of a single region $R$ is given by
$$A_R(b_R) = \sum_{x_R} b_R(x_R) \left(E_R(x_R) + \sum_{j\prec R} E_j(x_j)\right) + kT \sum_{x_R} b_R(x_R) \log b_R(x_R).$$
Unless $R$ is isolated from the rest of the system one would not expect that minimizing $A_R$ would produce a result related to our desired posterior distribution. Nonetheless, we argue as follows. At the minimum of the free energy, the outgoing message $M_{j\to S}(x_j) \propto e^{-kT\lambda_{jS}(x_j)}$ represents the flow out of variable $j$ into region $S$. The (change in) free energy associated to this is 
$$A_j^{S}(b_j) = \sum_{x_j} b_j(x_j) \lambda_{jS}(x_j) + kT b_j(x_j)\log b_j(x_j).$$
At each variable in $\partial R$ we compensate for this expected flow. This produces the \emph{modified free energy} of the region $R$, defined as
$$A_R(b_R) - \sum_{j\in\partial R}\left( \sum_{S\succ j\::\:S\not= R} A_j^{S}(b_j)\right).$$
As with the Bethe approximation, we add appropriate marginalization constraints, to obtain the \emph{constrained modified free energy}
\begin{eqnarray*}
\tilde{A}_R &=& \sum_{x_R} b_R(x_R) \tilde{E}_R(x_R) + kT \sum_{x_R} b_R(x_R) \log b_R(x_R)\\
&&\quad +\ \sum_{j\prec \partial R}\sum_{x_j} b_j(x_j) V^R_j(x_j) - kT(C_j-1) \sum_{x_j} b_j(x_j) \log b_j(x_j)\\
&&\quad +\ \lambda_R\left(\sum_{x_R} b_R(x_R) - 1 \right) + \sum_{j\prec \partial R} \lambda^R_j \left(\sum_{x_j} b_j(x_j) - 1\right)\\
&&\quad +\ \sum_{j\prec \partial R} \sum_{x_j}\lambda_{jR}(x_j) \left( \sum_{x_R\setminus x_j} b_R(x_R) - b_j(x_j)\right),
\end{eqnarray*}
where
$$\tilde{E}_R(x_R) = E_R(x_R) + \sum_{j \prec R\setminus\partial R} E_j(x_j)$$
and for $j\in\partial R$,
$$V_j^R(x_j) = E_j(x_j) - \sum_{S \succ j \::\: S\not= R} \lambda_{jS}(x_j).$$
Note the the definition of $V_j^R$ involves the Lagrange multipliers of the constrained modified free energy of adjacent regions.

Adding these local ``corrective'' potentials $V_j^R$ could be viewed as a form of hybrid mean field approach \cite{riegler2013merging}, however the values of the potential are unknown. Nonetheless, like the cavity method \cite{mooij2007loop}, if the corrective potentials are set as indicated above then the minima of each regional free energy do minimize the global approximate free energy.

\begin{theorem}\label{theorem:reconstitute}
A joint interior critical point of each modified regional free energy defines a critical point of the regional approximate free energy.
\end{theorem}
\begin{proof}
Just as above we compute an interior critical point of each region's constrained modified free energy, yielding
\begin{eqnarray*}
\tilde{E}_R(x_R) + kT(\log b_R(x_R) + 1) + \lambda_R + \sum_{j\prec \partial R} \lambda_{jR}(x_j) &=& 0,\\
V_j^R(x_j) - (C_j-1)kT(\log b_j(x_j)  + 1) + \lambda^R_j - \lambda_{jR}(x_j) &=& 0.
\end{eqnarray*}
Now using $V_j^R(x_j) = E_j(x_j) - \sum_{S \succ j \::\: S\not= R} \lambda_{jS}(x_j)$ these become
\begin{eqnarray*}
b_R(x_R) &\propto& \tilde{f}_R(x_R)\cdot \prod_{j\prec \partial R} e^{-\lambda_{jR}(x_j)/kT}\\
b_j(x_j) &\propto & f(x_j)^{1/(1-C_j)}\cdot \prod_{S\succ j} e^{-\lambda_{jS}(x_r)/kT(C_j-1)}.
\end{eqnarray*}
These are precisely the relations of a critical point of the regional approximate free energy.
\end{proof}

Finally we can tackle the question of how to utilize a modest sized device or method that can produce Boltzmann samples. \emph{Suppose that}: given an arbitrary selection of corrective potentials $\{V_j^R(x_j)\}_{j \prec \partial R}$, we have a black box that produces an interior minimum $b_R(x_R)$ of each region's modified regional free energy. \emph{Then}: we instantiate Algorithm \ref{algorithm} below. Note that in this algorithm it is only the marginals of the Boltzmann distribution $b_R(x_R)$ that are required. With the ability to do Boltzmann sampling from physical hardware (or by other means), we approximate these marginals in each region by drawing a sufficiently large number of samples and estimating the marginal probabilities empirically. Specifically, we decompose our large re-estimation problem into regions small enough for our Boltzmann sampler to handle, and use the sampler on each region to approximate the marginals of Algorithm \ref{algorithm}. By Theorem \ref{theorem:soundness} below, if the algorithm converges then it recovers the approximate posterior.

\begin{algorithm}[h]
  Initialize each $F_{j\to R}(x_j)$ appropriately (e.g. to a uniform distribution) for each $j,R$ with $j \prec \partial R$\;
  \Repeat{converged or timed out}{
    compute the potentials $V_j^R(x_j) = E_j(x_j) - \sum_{S\succ j\::\: S\not= R} kT \log F_{j\to S}(x_j)$\;
    with these potentials, use the black box to obtain free energy minima $b_R(x_R)$\;
    compute the messages $F_{R\to j}(x_j) \propto F_{j\to R}(x_j)^{-1}\cdot \sum_{x_R\setminus x_j} b_R(x_R)$\;
    re-estimate $F_{j\to R}(x_j) \propto f_j(x_j)\cdot \prod_{S\succ j\::\: S\not= R} F_{S\to j}(x_j)\;$
  }
  \caption{The domain decomposition algorithm}\label{algorithm}
\end{algorithm}

The numerical results presented in \cite{bian2014discrete, bian2016mapping} indicate that these approximations can be sufficient to solve problems too large to be sampled directly. In \cite{bian2014discrete} we used a variant of Algorithm \ref{algorithm} based on the min-sum method to solve 1000 variable LDPC decoding problems on a 504 qubit D-Wave quantum annealer for which belief propagation failed to converge. In \cite{bian2016mapping} we implemented a slight variant of Algorithm \ref{algorithm} to solve hardware fault diagnosis problems requiring up to seven regions to completely embed on a D-Wave 2X architecture.

\begin{lemma}
If Algorithm \ref{algorithm} converges then the marginal at any boundary variable $b_j(x_j) = \sum_{x_R\setminus x_j} b_R(x_R)$ is independent of the region $R$ used to compute it.
\end{lemma}
\begin{proof}
For any boundary variable $x_j$ and region $R \succ j$, if the algorithm has converged then steps 5 and 6 of the algorithm show
$$\sum_{x_R\setminus x_j} b_R(x_R) \propto F_{j\to R}(x_j)\cdot F_{R\to j}(x_j) \propto f_j(x_j)\cdot \prod_{S\succ j} F_{S\to j}(x_j),$$
which is independent of $R$.
\end{proof}

\begin{theorem}\label{theorem:soundness}
If Algorithm \ref{algorithm} converges then the collection of minima of each region's modified regional free energy $b_R(x_R)$ from step 4, and their marginals $b_j(x_j) = \sum_{x_R\setminus x_j} b_R(x_R)$, produces a critical point of the approximate regional free energy of the whole system.
\end{theorem}
\begin{proof}
For an arbitrary selection of corrective potentials $V_j^R(x_j)$, the critical point of the constrained modified free energy satisfies
\begin{eqnarray*}
b_R(x_R) &\propto& \tilde{f}_R(x_R)\cdot \prod_{j\prec \partial R} e^{-\lambda_{jR}(x_j)/kT}\\
b_j^{(R)}(x_j) &\propto & e^{V_j^R(x_j)/kT(C_j-1)} \cdot e^{-\lambda_{jR}(x_j)/kT(C_j-1)}.
\end{eqnarray*}
Therefore at a critical point,
$$b_R(x_R) \propto \tilde{f}_R(x_R)\cdot \prod_{j\prec \partial R} b^{(R)}_j(x_j)^{(C_j-1)} e^{-V_j^R(x_j)/kT}.$$
Define
\begin{eqnarray*}
M_{j\to R}(x_j) &\propto& b^{(R)}_j(x_j)^{(C_j-1)} e^{-V_j^R(x_j)/kT}\\
M_{R \to j}(x_j) &\propto& \sum_{x_R\setminus x_j} \tilde{f}_R(x_R)\cdot \prod_{k\prec \partial R\setminus j} M_{k\to R}(x_k).
\end{eqnarray*}
Note then than
$$b_j^{(R)}(x_j) = \sum_{x_R\setminus x_j} b_R(x_R) \propto M_{j\to R}(x_j) \cdot M_{R\to j}(x_j).$$

Now suppose that we are at a stationary point of the algorithm. Then from the lemma,
$$b_j(x_j)^{(C_j-1)} \propto \prod_{S\succ j\::\:S \not= R} F_{j\to S}(x_j)\cdot F_{S\to j}(x_j) \propto (f_j(x_j))^{-1}\cdot \prod_{S\succ j} F_{j\to S}(x_j).$$
Also,
$$e^{-V_j^R(x_j)/kT} = f_j(x_j) \cdot \prod_{S\succ j\::\: S \not=R} (F_{j\to S}(x_j))^{-1}.$$
Thus,
$$M_{j\to R}(x_j) \propto b^{(R)}_j(x_j)^{(C_j-1)} e^{-V_j^R(x_j)/kT} \propto F_{j\to R}(x_j).$$
But since both these are probabilities, this proportionality is in fact an equality. Also by the lemma,
$$F_{j\to R}(x_j)\cdot F_{R\to j}(x_j) \propto b_j(x_j) \propto M_{j\to R}(x_j) \cdot M_{R\to j}(x_j),$$
and so
$$M_{R\to j}(x_j) = F_{R\to j}(x_j).$$
Therefore, we must have
$$M_{j\to R}(x_j) \propto f_j(x_j) \cdot \prod_{S \succ j\::\:S\not=R} M_{S\to j}(x_j)$$
since this relation is satisfied by the $F$-messages. The result then follows from Proposition \ref{proposition:regional-BP}.
\end{proof}

\end{document}